\documentclass[runningheads]{llncs}
\usepackage[T1]{fontenc}
\usepackage{graphicx}
\usepackage{cite}
\usepackage{graphicx}
\usepackage{amsmath}
\usepackage{yhmath} 

\usepackage{amsfonts} 
\usepackage{amssymb}
\usepackage{color}
\usepackage{latexsym}

\usepackage{algorithm}
\usepackage[noend]{algorithmic}
\usepackage{multirow} 
\usepackage{amsmath}

\newcommand{\old}[1]{{}}

\newcommand{\T}{{{\cal{T}}}}

\begin{document}

\title{Geometric Spanning Trees Minimizing the Wiener Index}
%
%

\author{A.\,Karim Abu-Affash\inst{1} \and
Paz Carmi\inst{2} \and
Ori Luwisch\inst{2} \and
Joseph S. B. Mitchell\inst{3}
}
\authorrunning{A. K. Abu-Affash, P. Carmi, O. Luwisch, and J. S. B. Mitchell}
%
\institute{Department of Software Engineering, Shamoon College of Engineering, Israel \\ \email{abuaa1@sce.ac.il} \and
Computer Science Department, Ben-Gurion University, Israel \\ \email{carmip@cs.bgu.ac.il} \\ \email{orilu@post.bgu.ac.il} \and
Department of Applied Mathematics and Statistics, Stony Brook University, USA \\ \email{joseph.mitchell@stonybrook.edu}
}

\maketitle

\begin{abstract}
The Wiener index of a network, introduced by the chemist Harry Wiener~\cite{Wiener1947}, is the sum of distances between all pairs of nodes in the network.  
This index, originally used in chemical graph representations of the non-hydrogen atoms of a molecule, is considered to be a fundamental and useful network descriptor. 
We study the problem of constructing geometric networks on point sets in Euclidean space that minimize the Wiener index: given a set $P$ of $n$ points in $\mathbb{R}^d$, the goal is to construct a network, spanning $P$ and satisfying certain constraints, that minimizes the Wiener index among the allowable class of spanning networks.

In this work, we focus mainly on spanning networks that are trees and we focus on problems in the plane ($d=2$). We show that any spanning tree that minimizes the Wiener index has non-crossing edges in the plane. Then, we use this fact to devise an $O(n^4)$-time algorithm that constructs a spanning tree of minimum Wiener index for points in convex position.
We also prove that the problem of computing a spanning tree on $P$ whose Wiener index is at most $W$, while having total (Euclidean) weight at most $B$, is NP-hard.

Computing a tree that minimizes the Wiener index has been studied in the area of communication networks, where it is known as the \emph{optimum communication spanning tree problem}.

\end{abstract}

\section{Introduction}
The \emph{Wiener index} of a weighted graph $G=(V,E)$ is the sum, $\sum_{u,v \in V} \delta_G(u,v)$, of the shortest path lengths in the graph between every pair of vertices, where $\delta_G(u,v)$ is the weight of the shortest (minimum-weight) path between $u$ and $v$ in $G$. 
The Wiener index was introduced by the chemist Harry Wiener in 1947~\cite{Wiener1947}. The Wiener index and its several variations have found applications in chemistry, e.g., in predicting the antibacterial activity of drugs and modeling crystalline phenomena. It has also has been used to give insight into various chemical and physical properties of molecules \cite{Nenad86} and to correlate the structure of molecules with their biological activity \cite{kier2012molecular}. The Wiener index has become part of the general scientific culture, and it is still the subject of intensive research \cite{Bonchev02,Kinkar11,Dobrynin2001,Kexiang14}. In its applications in chemistry, the Wiener index is most often studied in the context of unweighted graphs. 
The study of minimizing the sum of interpoint distances also arises naturally in the network design field, where the problem of computing a spanning tree of minimum Wiener index is known as the \emph{Optimum Communication Spanning Tree} (OCST) problem~\cite{Hu1974,Gonzalez2007}. 

Given a undirected graph $G=(V,E)$ and a (nonnegative) weight function on the edges of $G$, representing the delay on each edge, the routing cost $c(T)$ of a spanning tree $T$ of $G$ is the sum of the weights (delays) of the paths in $T$ between every pair of vertices: $c(T)=\sum_{u,v \in V} \delta_T(u,v)$, where $\delta_T(u,v)$ is the weight of the (unique) path between $u$ and $v$ in $T$. 
The OCST problem aims to find a minimum routing cost spanning tree of a given weighted undirected graph $G$, thereby seeking to minimize the expected cost of a path within the tree between two randomly chosen vertices.
The OCST was originally introduced by Hu~\cite{Hu1974} and is known to be NP-complete in graphs, even if all edge weights are 1~\cite{Johnson1978}.
Wu et al.~\cite{Wu2000} presented a polynomial time approximation scheme (PTAS) for the OCST problem. Specifically, they showed that the best $k$-star (a tree with at most $k$ internal vertices) yields a $(\frac{k + 3}{k + 1})$-approximation for the problem, resulting in a $(1+\varepsilon)$-approximation algorithm of running time $O\big(n^{2\lceil\frac{2}{\varepsilon}\rceil-2}\big)$.

While there is an abundance of research related to the Wiener index, e.g., computing and bounding the Wiener indexes of specific graphs or classes of graphs \cite{GRAOVAC91,Harary69,Mohar88} and explicit formulas for the Wiener index for special classes of graphs \cite{Bonchev77,mekenyan1983,SHI93,Wiener1947,weiszfeldPointWhichSum2009}, to the best of our knowledge, the Wiener index has not received much attention in geometric settings. In this work, we study the Wiener index and the optimum communication spanning tree problem in selected geometric settings, hoping to bring this important and highly applicable index to the attention of computational geometry researchers.

\paragraph{Our Contributions and Overview.}
Let $P$ be a set of $n$ points in the plane. we study the problem of computing a spanning tree on $P$ that minimizes the Wiener index when the underlying graph is the complete graph on $P$, with edge weights given by their Euclidean lengths.  
In Section~\ref{sec:prelim}, we prove that the optimal tree (that minimizes the Wiener index) has no crossing edges. 
As our main algorithmic result, in Section~\ref{sec:dp}, we give a polynomial-time algorithm to solve the problem when the points $P$ are in convex position; this result strongly utilizes the structural result that the edges of an optimal tree do not cross, which enables us to devise a dynamic programming algorithm to optimize.
Then, in Section~\ref{sec:hardness}, we prove that the ``Euclidean Wiener Index Tree Problem'', in which we seek a spanning tree on $P$ whose Wiener index is at most $W$, while having total (Euclidean) weight at most $B$, is (weakly) NP-hard.  Finally, in Section~\ref{sec:paths}, we discuss the problem of finding a minimum Wiener index \emph{path} spanning $P$.  

\paragraph{Related Work.}
A problem related to ours is the minimum latency problem, also known as the traveling repairman problem TRP: Compute a path, starting at point $s$, that visits all points, while minimizing the sum of the distances (the ``latencies'') along the path from $s$ to every other point (versus between \emph{all} pairs of points, as in the Wiener index). There is a PTAS for TRP (and the $k$-TRP, with $k$ repairmen) in the Euclidean plane and in weighted planar graphs~\cite{sitters2021polynomial}.

Wiener index optimization also arises in the context of computing a noncontracting embedding of one metric space into another (e.g., a line metric or a tree metric) in order to minimize the average distortion of the embedding (defined to be the sum of all pairs distances in the new space, divided by the sum of all pairs distances in the original space). It is NP-hard to minimize average distortion when embedding a tree metric into a line metric; there is a constant-factor approximation (based on the $k$-TRP) for minimizing the average distortion in embedding a metric onto a line (i.e., finding a spanning path of minimum Wiener index)~\cite{dhamdhere2006approximation}, which, using \cite{sitters2021polynomial}, gives a $(2+\varepsilon)$-approximation in the Euclidean plane.

A related problem that has recently been examined in a geometric setting is the computation of the Beer index of a polygon $P$, defined to be the probability that two randomly (uniformly) distributed points in $P$ being visible to each other~\cite{abrahamsen2022degree}; the same paper also studies the problem of computing the expected distance between two random points in a polygon, which is, like the Wiener index, based on computing the sum of distances (evaluated as an integral in the continuum) between all pairs of points. 

Another area of research that is related to the Wiener index is that of \emph{spanners}: Given a weighted graph $G$ and a real number $t > 1$, a \emph{$t$-spanner} of $G$ is a spanning sub-graph $G^*$ of $G$, such that $\delta_{G^*}(u,v) \leq t \cdot \delta_G(u,v)$, for every two vertices $u$ and $v$ in $G$. Thus, the shortest path distances in $G^*$ approximate the shortest path distances in the underlying graph $G$, and the parameter $t$ represents the approximation ratio. The smallest $t$ for which $G^*$ is a $t$-spanner of $G$ is known as the \emph{stretch factor}. There is a vast literature on spanners, especially in geometry (see, e.g., \cite{BoseCC12,BoseGS02,BoseHS18,Cardinal04localproperties,FiltserS16,LiW04,narasimhanGeometricSpannerNetworks2007})
%
In a geometric graph, $G$, the \emph{stretch factor} between two vertices, $u$ and $v$, is the ratio between the Euclidean length of the shortest path from $u$ to $v$ in $G$ and the Euclidean distance between $u$ and $v$. The \emph{average stretch factor} of $G$ is the average stretch factor taken over all pairs of vertices in $G$.
For a given weighted connected graph $G=(V,E)$ with positive edge weights and a positive value $W$, the \emph{average stretch factor spanning tree} problem seeks a spanning tree $T$ of $G$ such that the average stretch factor (over $\binom n 2$ pairs of vertices) is bounded by $W$. For points in the Euclidean plane, 
one can construct in polynomial time a spanning tree with constant average stretch factor~\cite{DBLP:journals/jocg/ChengKLS12}. 

\old{ 
In typical research on geometric spanners, given a set of points and a constant $t > 1$, the goal is to construct a $t$-spanner such that the number of edges, the weight, and the degree are minimized (see \cite{BoseCC12,bose2005constructing,BoseHS18,Cardinal04localproperties,FiltserS16,LiW04} and the book \cite{narasimhanGeometricSpannerNetworks2007} by Narasimhan and Smid for a comprehensive overview on geometric spanner networks).
The geometric Wiener index can be viewed as a generalization of geometric spanners. In geometric spanners, we require that the shortest path distance between \emph{any pair} of points is at most some constant times their distance in the underlying graph, while in the Wiener index we consider the shortest path distances of all the pairs. A nice connection between spanners and the Wiener index is that a $t$-spanner $G^*$ of a graph $G$ has a Wiener index at most $t$ times the Wiener index of $G$. However, the inverse does not hold. For example, there exists a set of points in the plane such that there is no tree with a constant stretch factor, while there is always a tree with a Wiener index at most twice the Wiener index of the complete graph over the set of points.
} 
\old{ Joe commented out for now:
Problem definition: Given a set $P \subseteq \mathbb{R}^2$ of $n$ points, construct spanning tree $T$ of $P$ such that the Wiener index of $T$ is at most $(2-c)$ times the Wiener index of the complete Euclidean graph over $P$ for some constant $c>0$ in polynomial time.
Another related goal is to prove that for any $P \subseteq \mathbb{R}^2$ of $n$ points \emph{there exists} spanning tree $T$ such that the Wiener index of $T$ is at most $(2-c)$ times the Wiener index of the complete Euclidean graph over $P$ for some constant $c>0$.
That is, given a set $P \subseteq \mathbb{R}^2$ of $n$ points, the goal is to construct a graph $G = (P,E)$, such that $W(G)$ is as low as possible, in $O(poly(n))$ time under various constrains on the resulted graph $G$. The distances are computed by the Euclidean norm, that is, for $u,v \in P$ such that $(u,v) \in E$, $d_G(u,v) = |uv|$. Clearly, if there are no constraints on the resulted graph, the complete Euclidean graph is always gets the optimal (minimal) Wiener index. Thus, the research is about solving the problem with constrains on the resulted graph.
Most of the work described in this thesis is for the case where the resulted graph $G$ must be a \emph{tree}. In the general cases the exact value of the Wiener index of the optimal tree is unknown, therefore we will try to construct a tree such that its Wiener index approximates the Wiener index of the complete Euclidean graph.
As mentioned before, it is known (and will be discussed) how to construct a tree that its Wiener index is at most \emph{twice} the Wiener index of the complete graph over these points.
In this thesis, we also explain why solutions that immediately come to mind, like construction of minimum spanning tree, may provide a bad solutions in terms of Wiener index. We also present some other constructions that approximate the optimal solution and theoretical result as along with empirical results. We explain the cases where our constructions do not provide good enough  approximation and discuss other ways to get good Wiener spanning tree for those cases.
Another structures, where $G$ must be a \emph{path} or a \emph{planar} graph will be covered more shortly in section \ref{Other Structures}.
} 

    

\section{Preliminaries}\label{sec:prelim}
Let $P$ be a set of $n$ points in the plane and let $G=(P,E)$ be the complete graph over $P$.
For each edge $(p,q) \in E$, let $w(p,q)=|pq|$ denote the weight of $(p,q)$, given by the Euclidean distance, $|pq|$, between $p$ and $q$.
Let $T$ be a spanning tree of $P$.
For points $p,q \in P$, let $\delta_T(p,q)$ denote the weight of the (unique) path between $p$ and $q$ in $T$.
Let $W(T) = \sum_{p,q \in P} \delta_T(p,q)$ denote the Wiener index of $T$, given by the sum of the weights of the paths in $T$ between every pair of points.
Finally, for a point $p \in P$, let $\delta_p(T) = \sum_{q \in P}\delta_{T}(p,q)$ denote the total weight of the paths in $T$ from $p$ to every point of $P$. 

\old{
Reminder: Mention what we know about the degree of a node in opt Wiener tree (vs constant degree for MST).

Reminder: Mention hardness results known:  only the bicriteria case (weakly hard)?
Can we prove the 2D case is NP-hard in general?
(Mention that 1D case is trivial?)
The hardness in graphs is from Knapsack, so only weakly hard.
Open(?): is the network design problem strongly hard?

Reminder: Mention the model of computation allows summing square roots and comparing them.
}
 
\begin{theorem}
Let $T^*$ be a spanning tree of $P$ that minimizes the Wiener index. Then, $T^*$ is planar.
\end{theorem}
\begin{proof}
Assume towards a contradiction that there are two edges $(a,c)$ and $(b,d)$ in $T$ that cross each other.
Let $F$ be the forest obtained by removing the edges $(a,c)$ and $(b,d)$ from $T$. Thus $F$ contains three sub-trees. 
Assume, w.l.o.g., that $a$ and $b$ are in the same sub-tree $T_{ab}$, and $c$ and $d$ are in separated sub-trees $T_c$ and $T_d$, respectively; see Figure~\ref{fig:planarity}.
Let $n_{ab}$, $n_c$, and $n_d$ be the number of points in $T_{ab}$, $T_c$, and $T_d$, respectively. Thus,
\begin{align*}
                     W(T^*) & = W(T_{ab}) + n_c\cdot\delta_a(T_{ab}) + n_d\cdot\delta_b(T_{ab}) \\
                            & + W(T_c) + (n_{ab}+n_d)\cdot\delta_c(T_c) + n_c(n_{ab}+n_d)\cdot|ac| \\ 
                            & + W(T_d) + (n_{ab}+n_c)\cdot\delta_d(T_d) + n_d(n_{ab}+n_c)\cdot|bd| \\
                            & + n_c\cdot n_d \cdot \delta_{T^*}(a,b) \,.
\end{align*} 

\begin{figure}[htb]
  \centering
   \includegraphics[width=0.9\textwidth]{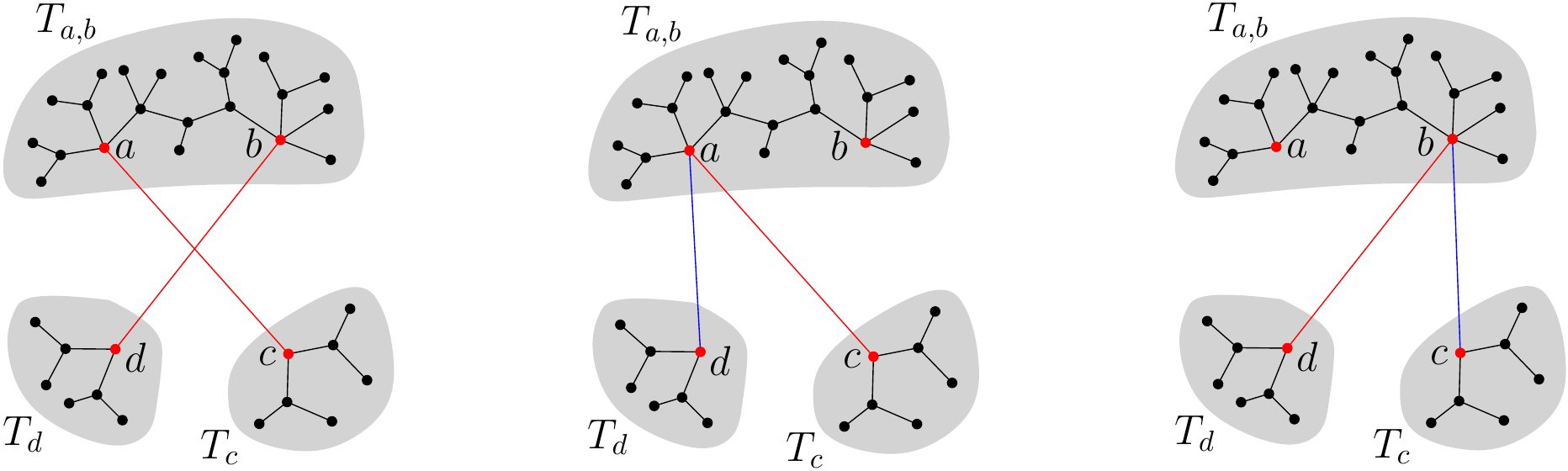}
    \caption{The trees $T^*$, $T'$, and $T''$ (from left to right).}
    \label{fig:planarity}
\end{figure}

Let $T'$ be the spanning tree of $P$ obtained from $T^*$ by replacing the edge $(b,d)$ by the edge $(a,d)$.
Similarly, let $T''$ be the spanning tree of $P$ obtained from $T^*$ by replacing the edge $(a,c)$ by the edge $(b,c)$.
Thus, 
\begin{align*}
                     W(T') & = W(T_{ab}) + (n_c + n_d)\cdot\delta_a(T_{ab}) \\
                            & + W(T_c) + (n_{ab}+n_d)\cdot\delta_c(T_c) + n_c(n_{ab}+n_d)\cdot|ac| \\ 
                            & + W(T_d) + (n_{ab}+n_c)\cdot\delta_d(T_d) + n_d(n_{ab}+n_c)\cdot|ad| \,,
\end{align*} 
and 
\begin{align*}
                     W(T'') & = W(T_{ab}) + (n_c + n_d)\cdot\delta_b(T_{ab}) \\
                            & + W(T_c) + (n_{ab}+n_d)\cdot\delta_c(T_c) + n_c(n_{ab}+n_d)\cdot|bc| \\ 
                            & + W(T_d) + (n_{ab}+n_c)\cdot\delta_d(T_d) + n_d(n_{ab}+n_c)\cdot|bd| \,.
\end{align*}
\old{
Therefore,
\begin{align*}
                     W(T^*) - W(T') & =  n_d\big(\delta_b(T_{ab}) - \delta_a(T_{ab})\big) + n_d(n_{ab}+n_c)\big(|bd|-|ad|\big) 
                             + n_c\cdot n_d \cdot \delta_{T^*}(a,b) \,,
\end{align*}
and
\begin{align*}
                     W(T^*) - W(T'') & =  n_c\big(\delta_a(T_{ab}) - \delta_b(T_{ab})\big) + n_c(n_{ab}+n_d)\big(|ac|-|bc|\big) 
                             + n_c\cdot n_d \cdot \delta_{T^*}(a,b) \,.
\end{align*}
If $W(T^*) - W(T') > 0$ or $W(T^*) - W(T'') > 0$, then this contradicts the minimality of $T^*$, and we are done.
Assume that $W(T^*) - W(T') \le 0$ and $W(T^*) - W(T'') \le 0$. Since $n_c > 0$ and $n_d > 0$, we have
\begin{align*}
                     \delta_b(T_{ab}) - \delta_a(T_{ab}) + (n_{ab}+n_c)\big(|bd|-|ad|\big) + n_c\cdot \delta_{T^*}(a,b) \le 0 \,,
\end{align*}
and
\begin{align*}
                     \delta_a(T_{ab}) - \delta_b(T_{ab}) + (n_{ab}+n_d)\big(|ac|-|bc|\big) + n_d \cdot \delta_{T^*}(a,b) \le 0 \,.
\end{align*} 
Thus, by summing these inequalities, we have
\begin{align*}
                     (n_{ab}+n_c)\big(|bd|-|ad|\big) + (n_{ab}+n_d)\big(|ac|-|bc|\big) + (n_c+n_d)\cdot \delta_{T^*}(a,b) \le 0  \,.
\end{align*}
That is, 
\begin{align*}
                     n_{ab}\big(|bd| + |ac| - |ad| - |bc|\big) &+ n_c\big(|bd| + \delta_{T^*}(a,b) - |ad| \big)  
										+ n_d\big(|ac| + \delta_{T^*}(a,b) - |bc| \big) \le 0  \,.
\end{align*}
Since $n_{ab}, n_c, n_d > 0$, and, by the triangle inequality, $|bd| + |ac| - |ad| - |bc| > 0$, $|bd| + \delta_{T^*}(a,b) - |ad| > 0$, and $|ac| + \delta_{T^*}(a,b) - |bc| > 0$, this is a contradiction.
}
Therefore,
\begin{align*}
                     W(T^*) - W(T') & =  n_d\big(\delta_b(T_{ab}) - \delta_a(T_{ab})\big) + n_d(n_{ab}+n_c)\big(|bd|-|ad|\big) \\
                            & + n_c\cdot n_d \cdot \delta_{T^*}(a,b) \,,
\end{align*}
and
\begin{align*}
                     W(T^*) - W(T'') & =  n_c\big(\delta_a(T_{ab}) - \delta_b(T_{ab})\big) + n_c(n_{ab}+n_d)\big(|ac|-|bc|\big) \\
                            & + n_c\cdot n_d \cdot \delta_{T^*}(a,b) \,.
\end{align*}
If $W(T^*) - W(T') > 0$ or $W(T^*) - W(T'') > 0$, then this contradicts the minimality of $T^*$, and we are done.

Assume that $W(T^*) - W(T') \le 0$ and $W(T^*) - W(T'') \le 0$. Since $n_c > 0$ and $n_d > 0$, we have
\begin{align*}
                     \delta_b(T_{ab}) - \delta_a(T_{ab}) + (n_{ab}+n_c)\big(|bd|-|ad|\big) + n_c\cdot \delta_{T^*}(a,b) \le 0 \,,
\end{align*}
and
\begin{align*}
                     \delta_a(T_{ab}) - \delta_b(T_{ab}) + (n_{ab}+n_d)\big(|ac|-|bc|\big) + n_d \cdot \delta_{T^*}(a,b) \le 0 \,.
\end{align*} 
Thus, by summing these inequalities, we have
\begin{align*}
                     (n_{ab}+n_c)\big(|bd|-|ad|\big) + (n_{ab}+n_d)\big(|ac|-|bc|\big) + (n_c+n_d)\cdot \delta_{T^*}(a,b) \le 0  \,.
\end{align*}
That is, 
\begin{align*}
                     n_{ab}\big(|bd| + |ac| - |ad| - |bc|\big) &+ n_c\big(|bd| + \delta_{T^*}(a,b) - |ad| \big) \\ 
										&+ n_d\big(|ac| + \delta_{T^*}(a,b) - |bc| \big) \le 0  \,.
\end{align*}
Since $n_{ab}, n_c, n_d > 0$, and, by the triangle inequality, $|bd| + |ac| - |ad| - |bc| > 0$, $|bd| + \delta_{T^*}(a,b) - |ad| > 0$, and $|ac| + \delta_{T^*}(a,b) - |bc| > 0$, this is a contradiction.~\qed
\end{proof}

\section{An Exact Algorithm for Points in Convex Position}\label{sec:dp}
Let $\{p_1,p_2,\ldots,p_n\}$ denote the vertices of the convex polygon that is obtained by connecting the points in $P$, ordered in clockwise-order with an arbitrary first point $p_1$; see Figure~\ref{fig:DP}. 
For simplicity of presentation, we assume that all indices are taken modulo $n$.
For each $1 \le i \le j \le n$, let $P[i,j] \subseteq P$ be the set $\{p_i,p_{i+1},\dots ,p_j\}$.
Let $T_{i,j}$ be a spanning tree of $P[i,j]$, and let $W(T_{i,j})$ denote its Wiener index. 
For a point $x\in\{i,j\}$, let $\delta_x(T_{i,j})$ be the total weight of the shortest paths from $p_x$ to every point of $P[i,j]$ in $T_{i,j}$. That is $\delta_x(T_{i,j}) = \sum_{p \in P[i,j]}\delta_{T_{i,j}}(p_x,p)$.
\begin{figure}[htp]
    \centering
        \includegraphics[width=0.58\textwidth]{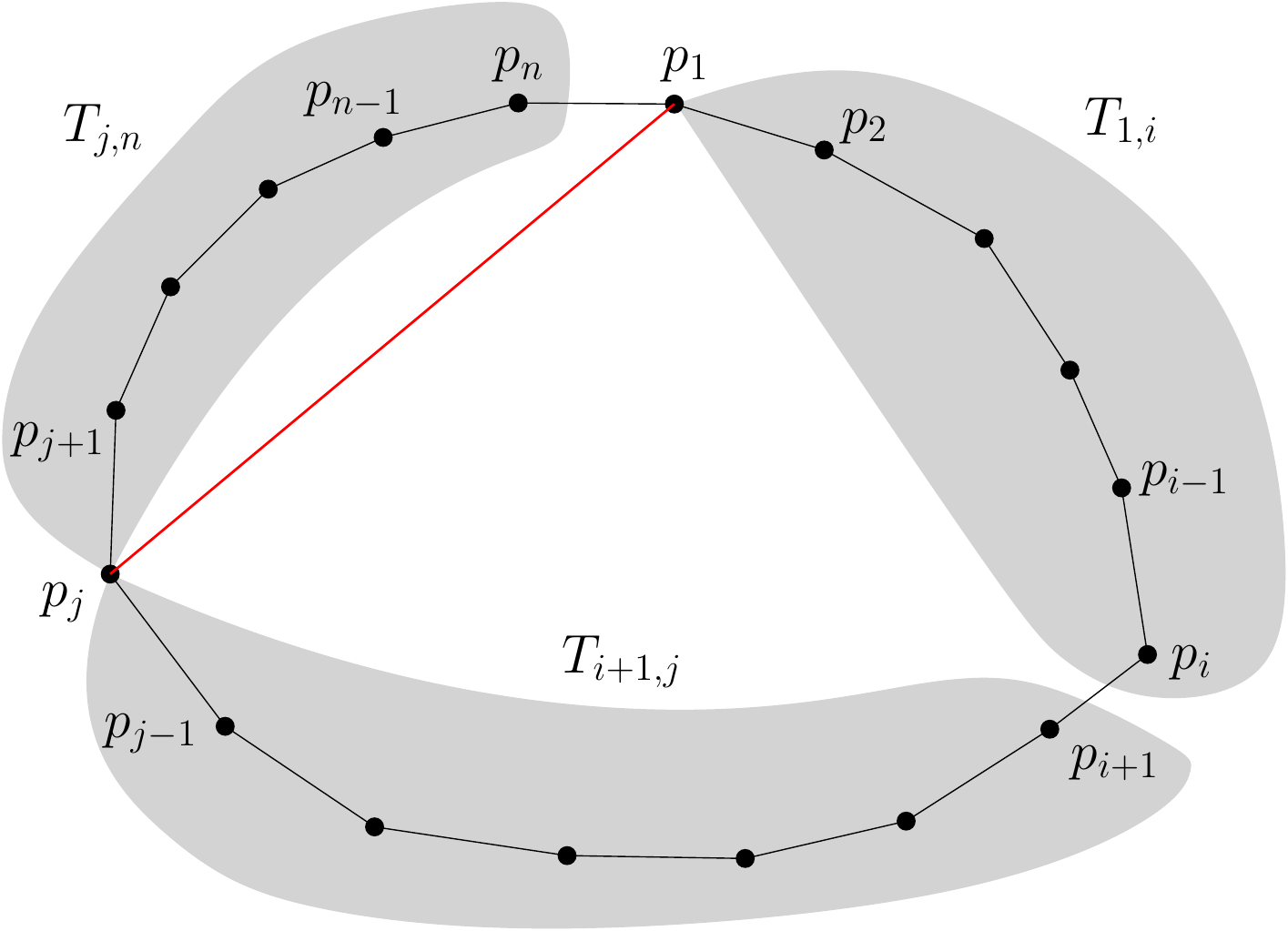}
    \caption{The convex polygon that is obtained from $P$. $p_1$ is connected to $p_j$ in $T^*$.}
    \label{fig:DP}
\end{figure}

Let $T^*$ be a minimum Wiener index tree of $P$ and let $W^*$ be its Wiener index. 
Notice that, for any $1\le i<j\le n$, the points in $P[i,j]$ are in convex position, since the points in $P$ are in convex position. 
Since $T^*$ is a spanning tree, each point, particularly $p_1$, is adjacent to at least one edge in $T^*$. 
Let $p_j$ be the point with maximum index $j$ that is connected to $p_1$ in $T^*$.
Moreover, there exists an index $1 \le i \le j$ such that all the points in $P[1,i]$ are closer to $p_1$ than to $p_j$ in $T^*$, and all the points in $P[i+1,j]$ are closer to $p_j$ than to $p_1$ in $T^*$.
Hence, 
\begin{align}
                     W^* & = W(T_{1,i}) + (n-i)\cdot\delta_1(T_{1,i}) \\
                            & + W(T_{i+1,j}) + (n-j+i)\cdot\delta_j(T_{i+1,j}) \\
                            & + W(T_{j,n}) + (j-1)\cdot\delta_j(T_{j,n}) \\
                            & + i(n-i)\cdot|p_1p_j|.
\end{align}

Thus, in order to compute $W^*$, we compute (1), (2), (3), and (4) for each $i$ between 2 and $n$ and for each $j$ between 1 and $i$, and take the minimum over the sum of these values. 
In general, for every $1\le i<j\le n$, let $W_j[i,j] = W(T_{i,j}) + (n-j+i-1)\cdot\delta_j(T_{i,j})$ be the minimum value obtained by a spanning tree $T_{i,j}$ of $P[i,j]$ rooted at $p_j$. Similarly, let $W_i[i,j] = W(T_{i,j}) + (n-j+i-1)\cdot\delta_i(T_{i,j})$ be the minimum value obtained by a spanning tree $T_{i,j}$ of $P[i,j]$ rooted at $p_i$. Thus, we can compute $W_j[i,j]$ and $W_i[i,j]$ recursively using the following formulas; see also Figure~\ref{fig:DP1}.

$$
W_j[i,j] =\underset{k \le l < j}{\underset{i \le k < j}{\min}} \ \big \{ W_k[i,k] + W_k[k,l] + W_j[l+1,j] + (l-i+1)(n-l+i-1)\cdot|p_kp_j| \big \} \,,
$$
and
$$
W_i[i,j] =\underset{i \le l < k}{\underset{i < k \le j}{\min}} \ \big \{ W_i[i,l] + W_k[l+1,k] + W_j[k,j] + (j-l)(n-j+l)\cdot|p_ip_k|\big \} \,.
$$
\begin{figure}[htp]
    \centering
        \includegraphics[width=0.8\textwidth]{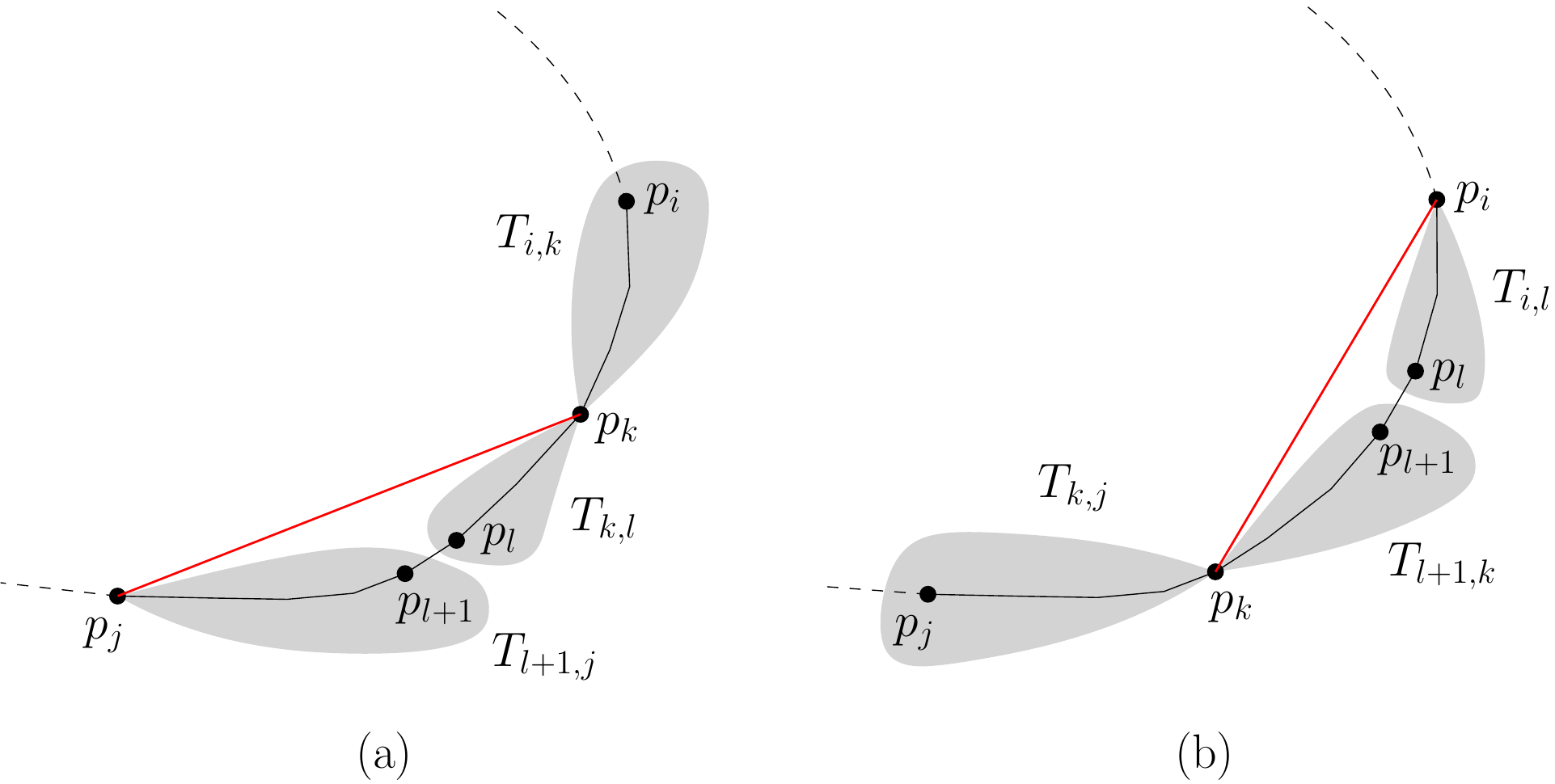}
    \caption{A sub-problem defined by $P[i,j]$. (a) Computing $W_j[i,j]$. (b) Computing $W_i[i,j]$.}
    \label{fig:DP1}
\end{figure}

We compute $W_j[i,j]$ and $W_i[i,j]$, for each $1 \le i < j \le n$, using dynamic programming as follows. 
We maintain two tables $\stackrel{\rightarrow}{M}$ and $\stackrel{\leftarrow}{M}$ each of size $n \times n$, such that $\stackrel{\rightarrow}{M}[i,j] = W_j[i,j]$ and $\stackrel{\leftarrow}{M}[i,j] = W_i[i,j]$, for each $1 \le i < j \le n$. 
We fill in the tables using Algorithm~\ref{algo:DP}.

\floatname{algorithm}{Algorithm}
\begin{algorithm}[ht]
\caption{$ComputeOptimal(P)$ } \label{algo:DP}
\begin{algorithmic}[1]

\STATE $n \leftarrow |P|$
\STATE \textbf{for} each $i \leftarrow 1$ to $n$ \textbf{do} \\ 
\quad  $\stackrel{\rightarrow}{M}[i,i] \leftarrow 0$ \\
\quad  $\stackrel{\leftarrow}{M}[i,i] \leftarrow 0$ \\

\STATE \textbf{for} each $j \leftarrow n$ to $1$ \textbf{do} \\
\quad   \textbf{for} each $i \leftarrow j$ to $n$ \textbf{do} \\
	            \quad \quad \ $\stackrel{\rightarrow}{M}[i,j] \leftarrow \underset{k \le l < j}{\underset{i \le k < j}{\min}} \ \big \{ \stackrel{\rightarrow}{M}[i,k] + \stackrel{\leftarrow}{M}[k,l] + \stackrel{\rightarrow}{M}[l+1,j] + (l-i+1)(n-l+i-1)\cdot|p_kp_j| \big \}$ \\
                \quad \quad \ $\stackrel{\leftarrow}{M}[i,j] \leftarrow \underset{i \le l < k}{\underset{i < k \le j}{\min}} \ \big \{ \stackrel{\leftarrow}{M}[i,l] + \stackrel{\rightarrow}{M}[l+1,k] + \stackrel{\rightarrow}{M}[k,j] + (j-l)(n-j+l)\cdot|p_ip_k|\big \}$ \\

\STATE \textbf{return} $\stackrel{\leftarrow}{M}[1,n]$
	
\end{algorithmic}
\end{algorithm}

Notice that when we fill the cell $\stackrel{\rightarrow}{M}[i,j]$, all the cells $\stackrel{\rightarrow}{M}[i,k]$, $\stackrel{\leftarrow}{M}[k,l]$, and $\stackrel{\rightarrow}{M}[l+1,j]$, for each $i \le k < j$ and for each $k \le l < j$, are already computed, and when we fill the cell $\stackrel{\leftarrow}{M}[i,j]$, all the cells $\stackrel{\leftarrow}{M}[i,l]$, $\stackrel{\rightarrow}{M}[l+1,k]$, and $\stackrel{\rightarrow}{M}[k,j]$, for each $i < k \le j$ and for each $i \le l < k$, are already computed. Therefore, each cell in the table is computed in $O(n^2)$ time, and the whole table is computed in $O(n^4)$ time.

The following theorem summarizes the result of this section.
\begin{theorem} \label{thm:opt}
Let $P$ be a set of $n$ points in convex position. Then, a spanning tree of $P$ of minimum Wiener index can be computed in $O(n^4)$ time.
\end{theorem}

\section{Hardness Proof  }\label{sec:hardness}
Let $P$ be a set of points in the plane and let $T$ be a spanning tree of $P$.
We define the Wiener index of $T$ as $W(T) = \sum_{p,q \in P} \delta_T(p,q)$ and the weight of $T$ as $wt(T)=\sum_{(p,q) \in T} |pq|$, where $\delta_T(p,q)$ is the length of the path between $p$ and $q$ in $T$ and $|pq|$ is the Euclidean distance between $p$ and $q$. For a edge $(p,q)$, let $N_T(p)$ (resp., $N_T(q)$) be the number of points in $T$ that are closer to $q$ than $q$ (resp., to $q$ than $p$). It is well known~\cite{Knor2015} that $W(T)$ can be formulated as: 
$$W(T) = \sum_{(p,q) \in T} N_T(p)\cdot N_T(q)\cdot |pq|.$$

In this section, we prove that the following problem is NP-hard. \\
\paragraph{\textbf{Euclidean Wiener Index Tree Problem: }} Given a set $P$ of points in the plane, a cost $W$, and a budget $B$, decide whether there exists a spanning tree $T$ of $P$, such that $W(T) \le W$ and $wt(T) \le B$.
\begin{theorem}
The Euclidean Wiener Index Tree Problem is weakly NP-hard.
\end{theorem}
\begin{proof}
Inspired by Carmi and Chaitman-Yerushalmi~\cite{Carmi13}, we reduce the Partition problem, which is known
to be NP-hard~\cite{Garey1979}, to the Euclidean Wiener Index Tree Problem. 
In the Partition problem, we are given a set $X=\{x_1,x_2,\ldots,x_n\}$ of $n$ positive integers with even $R=\sum_{i=1}^{n}x_i$, and the goal is to decide whether there is a subset $S\subseteq X$, such that $\sum_{x_i \in S} x_i = \frac{1}{2}R$.

Given an instance $X=\{x_1,x_2,\ldots,x_n\}$ of the Partition problem, where $x_i$'s are integers, 
we construct a set $P$ of $m = n^3 + 3n$ points as follows.
The set $P$ consists of $n$ points $p_1,p_2,\ldots,p_n$ located equally spaced on a circle of radius $nR$, a cluster $C$ of $n^3$ points located on the center of the circle.
Moreover, for each $1 \le i \le n$, we locate two points $l_i$ and $r_i$ both of distance $x_i$ from $p_i$ and the distance between them is $\frac{1}{2}x_i$; see Figure~\ref{fig:treeReduction}.
Finally, we set 
\begin{align*}
B=& \ \Big( n^2 + \frac{7}{4}\Big)R \text{, and  } \\ 
W=& \ 3n^2\big( m-3 \big)R + \Big( \frac{9}{4}m - \frac{13}{4}\Big)R \\
 =& \ 3n^5R + \frac{45}{4}n^3R -9n^2R + \frac{27}{4}nR - \frac{13}{4}R \,. 
\end{align*}
\begin{figure}[H]
    \centering
        \includegraphics[width=0.7\textwidth]{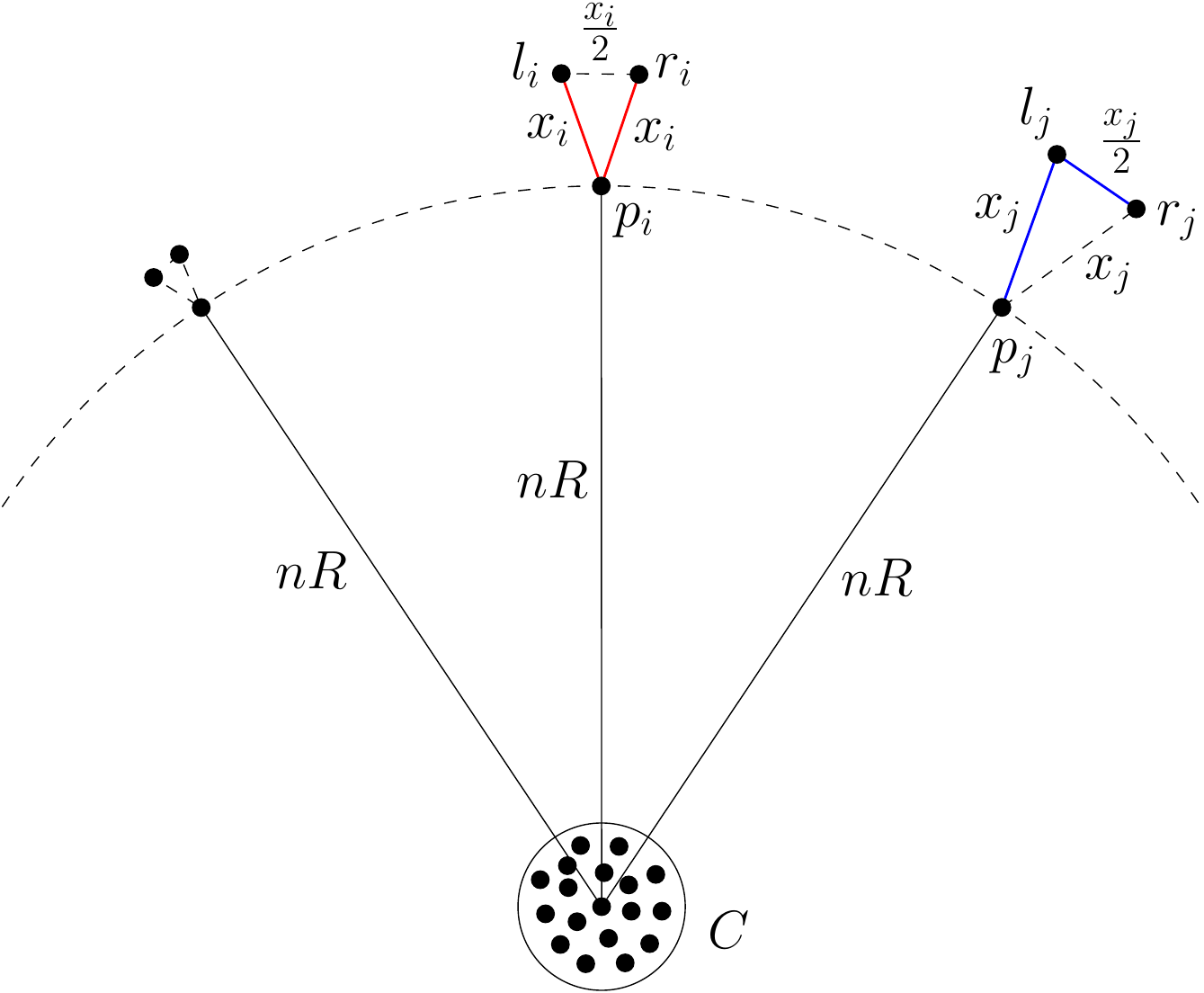}
    \caption{The set $P$ produced by the reduction. Connecting the points $l_j$, $r_j$, and $p_j$ for $x_j \in S$ (blue) and connecting the points $l_i$, $r_i$, and $p_i$ for $x_i \in  X\setminus S$ (red).}
    \label{fig:treeReduction}
\end{figure}

Assume that there exists a set $S\subseteq X$, such that $\sum_{x_i \in S} x_i = \frac{1}{2}R$.
We construct a spanning tree $T$ for the points in $P$ as follows:
\begin{itemize}
    \item Select an arbitrary point $s \in C$ and connect it to all the points in $C\cup\{p_1,p_2,\ldots,p_n\}$ as a star centered at $s$.b
    \item For each $1 \le i \le n$, connect the points $p_i$ and $l_i$.
    \item For each $x_i \in S$, connect the points $p_i$ and $r_i$.
    \item For each $x_i \in X\setminus S$, connect the points $r_i$ and $l_i$; see Figure~\ref{fig:treeReduction}.
\end{itemize}
It is easy to see that $wt(T)= n^2R + R + \frac{3}{4}R = \big( n^2 + \frac{7}{4}\big)R = B$. 
Moreover, the Wiener index of $T$ is:
\begin{align*}
 W(T) = & \sum_{(p,q) \in T} N_T(p)\cdot N_T(q)\cdot |pq|  \\ 
      = & \ 3 (n^3 + 3n -3) n^2 R +  \sum_{x_i \in S'} 2(n^3 + 3n -1) x_i  \\ 
            & \quad + \sum_{x_i \notin S'} \Big( (n^3 + 3n -1) \frac{1}{2}x_i\Big) +  \sum_{x_i \notin S'} \Big( 2 (n^3 + 3n -2) x_i \Big)  \\
         = & \ 3 n^5 R + 9n^3 R -9n^2 R +  (n^3 + 3n -1)R  \\
           & \quad + \frac{1}{4}(n^3 + 3n -1) R + (n^3 + 3n -2) R \\ 
         = & \ 3n^5R + \frac{45}{4}n^3R -9n^2R + \frac{27}{4}nR - \frac{13}{4}R = W \,.
\end{align*}

Conversely, let $T'$ be a spanning tree of $P$ with $wt(T') \le B$ and  $W(T') \le W$.

\begin{claim}\label{cl:numberOfEdgesIsn}
The number of edges $(p, q) \in T'$, such that $p \in C$ and $q \in P \setminus C$ is $n$. 
\end{claim}
\begin{proof}
Assume there are $k$ such edges.
The weight of each such edge is at least $nR$ thus the $wt(T') \ge k nR$,  since $B = (n^2 + \frac{7}{4}) R$ we get that $k \le n$.
We have 
\begin{align*}
    W(T') > &  \ (3 k nR  +  3 (n-k)(nR + 2\pi R)) n^3 \\
          = & \ (3 k n  +  3 n^2  + 6n \pi - 3kn - 6 k \pi ) n^3 R \\
          = &  \ (    3 n^2  + 6 \pi (n  - k)) n^3 R \\
          = &  \  3 n^5 R   + 6 \pi (n  - k) n^3 R \,.
\end{align*}
Thus, if $ k < n $, then we get that $W(T') > 3 n^5 R   + 6 \pi  n^3 R > W$, for sufficiently large $n$. \qed
\end{proof}

Let $G_i = \{p_i, l_i, r_i \}$, for every $1 \le i \le n$. 
From the proof of Claim~\ref{cl:numberOfEdgesIsn}, if follows that 
for every $1 \le i \le n$, there is an exactly one edge $(p, q)$ in $T'$, where $q \in G_i$ and $p \in C$.
Moreover, it is easy to see that $q = p_i$. 
Thus, in every $G_i$, we have $(p_i, l_i) \in T'$ or  $(p_i, r_i) \in T'$. 
 Assume w.l.o.g., that $(p_i, l_i) \in T'$.
Therefore, either $(p_i, r_i) \in T'$ or $(l_i, r_i) \in T'$.
Let $S' \subseteq X$, such that $x_i \in S'$ if and only if $(p_i,r_i) \in T'$, and let  $R' =\sum_{x_i \in S'} x_i$. 

Thus, to finish the proof we show that if  $R' \neq \frac{1}{2}R$, then
either $wt(T') > B$ or $W(T) > W$.

\noindent
\textbf{Case 1:} $R' > \frac{1}{2}R$. In this case, we have
\begin{align*}
    wt(T') \ge & \ n^{2} R + \sum_{x_i \in S'} 2 x_i + \sum_{x_i \notin S'} \frac{3}{2} x_i \  =   \  n^{2} R  + 2R' + \frac{3}{2}(R-R')  \\ 
    =& \ n^{2} R + \frac{1}{2}R' + \frac{3}{2}R  \  >  \ n^{2} R + \frac{1}{4}R + \frac{3}{2}R  \ =  \ \big( n^2 + \frac{7}{4}\big)R \ = B \,.
\end{align*}
Therefore,  $wt(T') > B$.

\noindent
\textbf{Case 2:} $R' < \frac{1}{2}R$. In this case, we have 
\begin{align*}
W(T) = & \sum_{(p,q) \in T} N_T(p)\cdot N_T(q)\cdot |pq|  \\ 
      = & \ 3 (n^3 + 3n -3) n^2 R +  \sum_{x_i \in S'} 2(n^3 + 3n -1) x_i  \\ 
            & \quad + \sum_{x_i \notin S'} \Big( (n^3 + 3n -1) \frac{1}{2}x_i\Big) +  \sum_{x_i \notin S'} \Big( 2 (n^3 + 3n -2) x_i \Big)  \\
         = & \ 3 n^5 R + 9n^3 R -9n^2 R +  2(n^3 + 3n -1)R'  \\
           & \quad + \frac{1}{2}\Big(n^3 + 3n -1\Big) (R-R') + 2(n^3 + 3n -2) (R-R') \\ 
\end{align*}
\begin{align*}
        = &  \ 3 n^5 R + 9n^3 R -9n^2 R + 2 (n^3 + 3n -2)R  \\
          & \quad  -  \Big(\frac{1}{2}\Big(n^3 + 3n -1\Big) -2\Big)R'  +  \frac{1}{2}\Big(n^3 + 3n -1\Big) R   \\ 
          & \quad  -  \Big(\frac{1}{2}\Big(n^3 + 3n -1\Big) -2\Big)R'  +  \frac{1}{2}\Big(n^3 + 3n -1\Big) R   \\             
         > &  \ 3 n^5 R + 9n^3 R - 9n^2 R + 2 (n^3 + 3n -2)R  \\ 
          & \quad  -  \frac{1}{2}\Big(\frac{1}{2}\Big(n^3 + 3n -1\Big) -2\Big)R  +  \frac{1}{2}\Big(n^3 + 3n -1\Big) R   \\  
        = & \ 3n^5R + \frac{45}{4}n^3R -9n^2R + \frac{27}{4}nR - \frac{13}{4}R = W \,.
\end{align*}

\old{ 
\begin{align*}
    W(T') = & \ 3 (n^3 + 3n -3) n^2 R +  \sum_{x_i \notin S'} (2 (n^3 + 3n -2) x_i)  \\ 
            & \quad                     + \sum_{x_i \in S'} ( (n^3 + 3n -1) x_i) +  \sum_{x_i \in S'} (  (n^3 + 3n -1) x_i)  \\
            & \quad   + \sum_{x_i \notin S'} ( (n^3 + 3n -1) \frac{1}{2} x_i) \\
          > &  \ (3 n^5  + \frac{45}{4}n^3  -9n^2  +  \frac{27}{4} n  -  \frac{13}{4} )  R  = W\\        
\end{align*}
Therefore,  $W(T') > W$. \qed


\begin{align*}
    W(T') = & \ 3 (n^3 + 3n -3) n^2 R +  \sum_{x_i \notin S'} (2 (n^3 + 3n -2) x_i)  \\ 
            & \quad                     + \sum_{x_i \in S'} ( (n^3 + 3n -1) x_i) +  \sum_{x_i \in S'} (  (n^3 + 3n -1) x_i)  \\
            & \quad   + \sum_{x_i \notin S'} ( (n^3 + 3n -1) \frac{1}{2} x_i) \\
         = & \ 3 n^5 R + 9n^3 R -9n^2 R +  2 (n^3 + 3n -2)(R- R')  \\ 
            & \quad  +   (n^3 + 3n -1) R' +   (n^3 + 3n -1)R'  \\
            & \quad   +  (n^3 + 3n -1) \frac{1}{2}(R- R') \\
         = &  \ 3 n^5 R + 9n^3 R -9n^2 R + 2 (n^3 + 3n -2)R- 2 (n^3 + 3n -2)R'  \\ 
            & \quad  +   (n^3 + 3n -1) R' +   (n^3 + 3n -1)R'  \\
            & \quad   +  \frac{1}{2}(n^3 + 3n -1)  R  -   \frac{1}{2}(n^3 + 3n -1)R' \\
         = &  \ 3 n^5 R + 9n^3 R -9n^2 R + 2 (n^3 + 3n -2)R  \\ 
          & \quad ( n^3 + 3n -1 -2 (n^3 + 3n -2)    +   n^3 + 3n -1  - \frac{1}{2}(n^3 + 3n -1) )R'  \\
            & \quad   +  \frac{1}{2}(n^3 + 3n -1)  R   \\   
        = &  \ 3 n^5 R + 9n^3 R -9n^2 R + 2 (n^3 + 3n -2)R  \\ 
          & \quad  -  (\frac{1}{2}(n^3 + 3n -1) -2)R'  +  \frac{1}{2}(n^3 + 3n -1)  R   \\             
         > &  \ 3 n^5 R + 9n^3 R - 9n^2 R + 2 (n^3 + 3n -2)R  \\ 
          & \quad  -  \frac{1}{2}(\frac{1}{2}(n^3 + 3n -1) -2)R  +  \frac{1}{2}(n^3 + 3n -1)  R   \\       
        = &  \ (3 n^5  + 11n^3  -9n^2  +  6n -4  - \frac{1}{2}(\frac{1}{2}(n^3 + 3n -1) -2)  +  \frac{1}{2}(n^3 + 3n -1))  R \\
        = &  \ (3 n^5  + \frac{45}{4}n^3  -9n^2  +  \frac{27}{4} n  -  \frac{13}{4} )  R \\    
\end{align*}
}

\old{
Given an instance $X=\{x_1,x_2,\ldots,x_n\}$ of the Partition problem, where $x_i$'s are integers, 
we construct a set $P$ of $2n^3 + 3n$ points as follows.
The set $P$ consists of two clusters $A$ and $B$ each of $n^3$ points. The points of $A$ are located 
on coordinate $(0,0)$ and the points of $B$ are located on coordinate $(R+n+1,0)$.
For each $1\le i \le n$, we locate a gadget $G_i$ consisting of three points $l_i$, $r_i$, and $p_i$, such that $l_i$ is on coordinate $(i + \sum_{i=1}^{i-1}x_i,0)$, $r_i$ is on $(i + \sum_{i=1}^{i}x_i , 0)$, and $p_i$ is on a positive $y$ coordinate of distance $\alpha x_i$ from both $l_i$ and $r_i$, where 
$\frac{1}{2} < \alpha < 1$, assume here $\alpha = \frac{2}{3}$; see Figure~\ref{fig:treeReduction}. 
We set 
\begin{align*}
B=& \ n+1 + \frac{3}{2}R \text{, and  } \\ 
W=& \ \big((n+1)R + \frac{7}{6}R\big)n^6 + 3n^5 R + 7 n^4 R \,.
\end{align*}
\begin{figure}[htp]
    \centering
        \includegraphics[width=0.98\textwidth]{treeReduction.pdf}
    \caption{Top: the set $P$ produced by the reduction. Bottom: connecting the points $l_i$, $r_i$, and $p_i$ when $x_i \in S$ (left) and when $x_i \in X\setminus S$ (right).}
    \label{fig:treeReduction}
\end{figure}

Assume that there exists a set $S\subseteq X$, such that $\sum_{x_i \in S} x_i = \frac{1}{2}R$.
We construct a spanning tree $T$ for the points in $P$ as follows:
\begin{itemize}
    \item Select an arbitrary point $a \in A$ and connect it to $l_1$ and to all the other points in $A$ as a star centered at $a$.
    \item Select an arbitrary point $b \in B$ and connect it to $r_n$ and to all the other points in $B$ as a star centered at $b$.
    \item For each $1 \le i \le n-1$, connect the points $r_i$ and $l_{i+1}$.
    \item For each $1 \le i \le n$, connect the points $l_i$ and $p_i$.
    \item For each $x_i \in S$, connect the points $l_i$ and $r_i$.
    \item For each $x_i \in X\setminus S$, connect the points $p_i$ and $r_i$; see Figure~\ref{fig:treeReduction}.
\end{itemize}
It is easy to see that $wt(T)=n+1 + \frac{2}{3}R + \frac{1}{2}R + \frac{2}{6}R = n+1 + \frac{3}{2}R = B$. 
Before computing the Wiener index of $T$, observe that: 
\begin{enumerate}
    \item $\delta_T(a,b) = n+1 + \frac{1}{2}R + \frac{4}{3}R = n+1 + \frac{7}{6}R$.
    \vspace{0.1cm}
    \item $ \sum_{i=1}^{n}(\delta_T(a,l_i) + \delta_T(l_i,b)) = n\delta_T(a,b) = n\big(n+1 + \frac{7}{6}R\big)$.
    \vspace{0.1cm}
    \item $ \sum_{i=1}^{n}(\delta_T(a,r_i), + \delta_T(r_i,b)) = n\delta_T(a,b) = n\big(n+1 + \frac{7}{6}R\big)$.
    \vspace{0.1cm}
    \item $ \sum_{i=1}^{n}(\delta_T(a,p_i) + \delta_T(p_i,b) )  =  n\delta_T(a,b) + \frac{2}{3}R =  n\big(n+1 + \frac{7}{6}R\big) + \frac{2}{3}R$.
    \vspace{0.2cm}
    \item $ \sum_{i=1}^{n} (\delta_T(p_i,l_i)+ \delta_T(p_i,r_i) + \delta_T(l_i,r_i) ) = (\frac{8}{3} + \frac{10}{3})\frac{1}{2}R= 3R$.
    \vspace{0.1cm}
    \item $ \sum_{i=1}^{n-1}\sum_{j=i+1}^{n}(\delta_T(q_i,q_j) ) \le n^2\delta_T(a,b) = n^2 \big(n+1 + \frac{7}{6}R\big) $, where $q_i \in \{l_i,r_i,p_i\}$ and $q_j \in \{l_j,r_j,p_j\}$.  
\end{enumerate}
Moreover, we have,
\begin{align*}  
W(T) =& \ n^6\delta_T(a,b)  \\ 
  +&  \ n^3 \sum_{i=1}^{n}(\delta_T(p_i,a) + \delta_T(l_i,a) + \delta_T(r_i,a)  
 + \delta_T(p_i,b) + \delta_T(l_i,b) + \delta_T(r_i,b) ) \\
 +& \ \sum_{i=1}^{n} (\delta_T(p_i,l_i)+ \delta_T(p_i,r_i) + \delta_T(l_i,r_i) ) \\
 +& \ \sum_{i=1}^{n-1}\sum_{j=i+1}^{n}(\delta_T(p_i,p_j) + \delta_T(p_i,l_j) + \delta_T(p_i,r_j)) \\
 +& \ \sum_{i=1}^{n-1}\sum_{j=i+1}^{n}(\delta_T(l_i,p_j) + \delta_T(l_i,l_j) + \delta_T(l_i,r_j)) \\
 +& \ \sum_{i=1}^{n-1}\sum_{j=i+1}^{n}(\delta_T(r_i,p_j) + \delta_T(r_i,l_j) + \delta_T(r_i,r_j)) 
  \end{align*}
 Thus, by plugging 1-6 we get 
 \begin{align*}  
W(T) \le & \  n^6\big(n+1 + \frac{7}{6}R\big) +  n^3 \Big(3n\big(n+1 + \frac{7}{6}R\big) +  \frac{2}{3}R\Big) + 3R \\
 & + 3n^2 \big(n+1 + \frac{7}{6}R\big) + 3n^2 \big(n+1 + \frac{7}{6}R\big) + 3n^2 \big(n+1 + \frac{7}{6}R\big)  \\
 = & \ n^6\big(n+1 + \frac{7}{6}R\big) + 3n^5 R + \frac{13}{2} n^4 R + \Theta(n^3)R \\
 \le & \  n^6\big(n+1 + \frac{7}{6}R\big) + 3n^5 R + 7 n^4 R \,.
  \end{align*}

Let $T'$ be a tree, and let $S' \subseteq X$, such that $x_i in S'$ if and only if 
$(l_i,r_i) \in T'$. 
To finish the proof we need to show that if  $\sum_{x_i \in S'} x_i \neq \frac{1}{2}R$, then
Either $wt(T') > B$ or $W(T) > W$.

\noindent
\textbf{Case 1:} $\sum_{x_i \in S'} x_i > \frac{1}{2}R$.
In this case, it is easy to see that $wt(T') > B = (n+\frac{5}{2})R$

\noindent
\textbf{Case 2:} $\sum_{x_i \in S'}x_i < \frac{1}{2}R$.
$\delta_T(a,b) >  \frac{1}{2}R -1   + (\frac{1}{2}R +1)\frac{4}{3}  + (n+1)R  =  \frac{7}{6} R + \frac{1}{3} + (n+1)R$

\begin{align*}  
W(T) =& \ n^6\delta_T(a,b)  \\ 
  +&  \ n^3 \sum_{i=1}^{n}(\delta_T(p_i,a) + \delta_T(l_i,a) + \delta_T(r_i,a)  
 + \delta_T(p_i,b) + \delta_T(l_i,b) + \delta_T(r_i,b) ) \\
 +& \ \sum_{i=1}^{n-1}\sum_{j=i+1}^{n}(\delta_T(p_i,p_j) + \delta_T(p_i,l_j) + \delta_T(p_i,r_j)) \\
 +& \ \sum_{i=1}^{n-1}\sum_{j=i+1}^{n}(\delta_T(l_i,p_j) + \delta_T(l_i,l_j) + \delta_T(l_i,r_j)) \\
 +& \ \sum_{i=1}^{n-1}\sum_{j=i+1}^{n}(\delta_T(r_i,p_j) + \delta_T(r_i,l_j) + \delta_T(r_i,r_j)) \\
 +& \ \sum_{i=1}^{n} (\delta_T(p_i,l_i)+ \delta_T(p_i,r_i) + \delta_T(l_i,r_i)    ) \\
 < &  \ n^6\delta_T(a,b)  +  n^3 \sum_{i=1}^{n}(\delta_T(p_i,a) + \delta_T(l_i,a) + \delta_T(r_i,a) 
  \end{align*}
}

\old{
Inspired by Carmi and Chaitman-Yerushalmi~\cite{Carmi13}, we reduce the Partition problem, which is known
to be NP-hard~\cite{Garey1979,Karp1972}, to the Euclidean Wiener Index Problem. 
In the Partition problem, we are given a set $X=\{x_1,x_2,\ldots,x_n\}$ of $n$ positive integers with even $R=\sum_{i=1}^{n}x_i$, and the goal is to decide whether there is a subset $S\subseteq X$, such that $\sum_{x_i \in S} x_i = \frac{1}{2}R$.

Given an instance $X=\{x_1,x_2,\ldots,x_n\}$ of the Partition problem, where $x_i$'s are integers, 
we construct a set $P$ of $2n^3 + 3n$ points as follows.
The set $P$ consists of two clusters $A$ and $B$ each of $n^3$ points. The points of $A$ are located 
on coordinate $(0,0)$ and the points of $B$ are located on coordinate $((n+2)R,0)$.
For each $1\le i \le n$, we locate three points $l_i$, $r_i$, and $p_i$, such that $l_i$ is on coordinate 
$(iR + \sum_{i=1}^{i-1}x_i,0)$, $r_i$ is on $(iR+ \sum_{i=1}^{i}x_i , 0)$, and $p_i$ is on a positive $y$ coordinate of distance $\alpha x_i$ from both $l_i$ and $r_i$, where 
$\frac{1}{2} < \alpha < 1$, assume $\frac{2}{3}$; see Figure~\ref{fig:treeReduction}. We set $B=(n+\frac{5}{2})R$ and $W=\big((n+1)R + \frac{7}{6}R\big)n^6 + 3n^5 R + 7 n^4 R$.
\begin{figure}[htp]
    \centering
        \includegraphics[width=0.98\textwidth]{treeReduction.pdf}
    \caption{Top: the set $P$ produced by the reduction. Bottom: connecting the points $l_i$, $r_i$, and $p_i$ when $x_i \in S$ (left) and when $x_i \in X\setminus S$ (right).}
    \label{fig:treeReduction}
\end{figure}

Assume that there exists a set $S\subseteq X$, such that $\sum_{x_i \in S} x_i = \frac{1}{2}R$.
We construct a spanning tree $T$ for the points in $P$ as follows:
\begin{itemize}
    \item Select an arbitrary point $a \in A$ and connect it to $l_1$ and to all the other points in $A$ as a star centered at $a$.
    \item Select an arbitrary point $b \in B$ and connect it to $r_n$ and to all the other points in $B$ as a star centered at $b$.
    \item For each $1 \le i \le n-1$, connect the points $r_i$ and $l_{i+1}$.
    \item For each $1 \le i \le n$, connect the points $l_i$ and $p_i$.
    \item For each $x_i \in S$, connect the points $l_i$ and $r_i$.
    \item For each $x_i \in X\setminus S$, connect the points $p_i$ and $r_i$; see Figure~\ref{fig:treeReduction}.
\end{itemize}
It is easy to see that $wt(T)=(n+1)R + \frac{2}{3}R + \frac{1}{2}R + \frac{4}{6}R = (n+\frac{5}{2})R = B$.
Before computing the Wiener index of $T$, observe that: 
\begin{enumerate}
    \item $\delta_T(a,b) = (n+1)R + \frac{1}{2}R + \frac{4}{3}R = (n+1)R + \frac{7}{6}R$.
    \item $ \sum_{i=1}^{n}(\delta_T(a,l_i) + \delta_T(l_i,b)) = n\delta_T(a,b) = n((n+1)R + \frac{7}{6}R)$.
    \item $ \sum_{i=1}^{n}(\delta_T(a,r_i), + \delta_T(r_i,b)) = n\delta_T(a,b) = n((n+1)R + \frac{7}{6}R)$.
    \item $ \sum_{i=1}^{n}(\delta_T(a,p_i) + \delta_T(p_i,b) )  =  n\delta_T(a,b) + \frac{2}{3}R =  n((n+1)R + \frac{7}{6}R) + \frac{2}{3}R$.
    \item $ \sum_{i=1}^{n} (\delta_T(p_i,l_i)+ \delta_T(p_i,r_i) + \delta_T(l_i,r_i) ) = (\frac{8}{3} + \frac{10}{3})\frac{1}{2}R= 3R$.
    \item $ \sum_{i=1}^{n-1}\sum_{j=i+1}^{n}(\delta_T(q_i,q_j) ) \le n^2\delta_T(a,b) = n^2 ((n+1)R + \frac{7}{6}R) $, where $q_i \in \{l_i,r_i,p_i\}$ and $q_j \in \{l_j,r_j,p_j\}$.  
\end{enumerate}

We have,
\begin{align*}  
W(T) =& \ n^6\delta_T(a,b)  \\ 
  +&  \ n^3 \sum_{i=1}^{n}(\delta_T(p_i,a) + \delta_T(l_i,a) + \delta_T(r_i,a)  
 + \delta_T(p_i,b) + \delta_T(l_i,b) + \delta_T(r_i,b) ) \\
 +& \ \sum_{i=1}^{n-1}\sum_{j=i+1}^{n}(\delta_T(p_i,p_j) + \delta_T(p_i,l_j) + \delta_T(p_i,r_j)) \\
 +& \ \sum_{i=1}^{n-1}\sum_{j=i+1}^{n}(\delta_T(l_i,p_j) + \delta_T(l_i,l_j) + \delta_T(l_i,r_j)) \\
 +& \ \sum_{i=1}^{n-1}\sum_{j=i+1}^{n}(\delta_T(r_i,p_j) + \delta_T(r_i,l_j) + \delta_T(r_i,r_j)) \\
 +& \ \sum_{i=1}^{n} (\delta_T(p_i,l_i)+ \delta_T(p_i,r_i) + \delta_T(l_i,r_i)    )
  \end{align*}
 Thus, by plugging 1-6 we get 
 \begin{align*}  
W(T) \le & \  n^6(\frac{7}{6}R + (n+1)R) \\ 
  +&  \  n^3 (3n(\frac{7}{6}R + (n+1)R ) +  \frac{2}{3}R) \\
 +& \ 3n^2 (\frac{7}{6}R + (n+1)R) \\
 +& \ 3n^2 (\frac{7}{6}R + (n+1)R)\\
 +& \ 3n^2 (\frac{7}{6}R + (n+1)R) \\
 +& \ 3R \\
 =& \ n^6(\frac{7}{6}R + (n+1)R) + 3n^5 R + \frac{13}{2} n^4 R + \Theta(n^3)R \\
 \le & \  n^6(\frac{7}{6}R + (n+1)R) + 3n^5 R + 7 n^4 R \,.
  \end{align*}

Let $T'$ be a tree, and let $S' \subseteq X$, such that $x_i in S'$ if and only if 
$(l_i,r_i) \in T'$. 
To finish the proof we need to show that if  $\sum_{x_i \in S'} x_i \noteq \frac{1}{2}R$, then
Either $wt(T') > B$ or $W(T) > W$.

\noindent
\textbf{Case 1:} $\sum_{x_i \in S'} x_i > \frac{1}{2}R$.
In this case, it is easy to see that $wt(T') > B = (n+\frac{5}{2})R$

\noindent
\textbf{Case 2:} $\sum_{x_i \in S'}x_i < \frac{1}{2}R$.
$\delta_T(a,b) >  \frac{1}{2}R -1   + (\frac{1}{2}R +1)\frac{4}{3}  + (n+1)R  =  \frac{7}{6} R + \frac{1}{3} + (n+1)R$

\begin{align*}  
W(T) =& \ n^6\delta_T(a,b)  \\ 
  +&  \ n^3 \sum_{i=1}^{n}(\delta_T(p_i,a) + \delta_T(l_i,a) + \delta_T(r_i,a)  
 + \delta_T(p_i,b) + \delta_T(l_i,b) + \delta_T(r_i,b) ) \\
 +& \ \sum_{i=1}^{n-1}\sum_{j=i+1}^{n}(\delta_T(p_i,p_j) + \delta_T(p_i,l_j) + \delta_T(p_i,r_j)) \\
 +& \ \sum_{i=1}^{n-1}\sum_{j=i+1}^{n}(\delta_T(l_i,p_j) + \delta_T(l_i,l_j) + \delta_T(l_i,r_j)) \\
 +& \ \sum_{i=1}^{n-1}\sum_{j=i+1}^{n}(\delta_T(r_i,p_j) + \delta_T(r_i,l_j) + \delta_T(r_i,r_j)) \\
 +& \ \sum_{i=1}^{n} (\delta_T(p_i,l_i)+ \delta_T(p_i,r_i) + \delta_T(l_i,r_i)    ) \\
 < &  \ n^6\delta_T(a,b)  +  n^3 \sum_{i=1}^{n}(\delta_T(p_i,a) + \delta_T(l_i,a) + \delta_T(r_i,a) 
  \end{align*}

}

\end{proof}

\section{Paths that Optimize Wiener Index}\label{sec:paths}
We consider now the case of spanning paths that optimize the Wiener index.

\begin{theorem} \label{thm:OptNotPlanar}
Let $P$ be a set of $n$ points.
The path that minimizes the Wiener index among all Hamiltonian paths of $P$ is not necessarily planar.
\end{theorem}
\begin{proof}
Consider the set $P$ of $n=2m+2$ points in convex position as shown in Figure~\ref{fig:notPlanarPathNew}. 
The set $P$ consists of two clusters $P_l$ and $P_r$  and two points $p$ and $q$, where $|P_l|= |P_r|=m$.
The points in cluster $P_l$ are arbitrarily close to the origin $(0,0)$, and
the points in cluster $P_r$ are arbitrarily close to coordinate $(6,0)$. The point $p$ is located on coordinate $(5,1)$ and the
point $q$ is located on coordinate $(5,-1)$.
\begin{figure}[htp]
    \centering
        \includegraphics[width=0.5\textwidth]{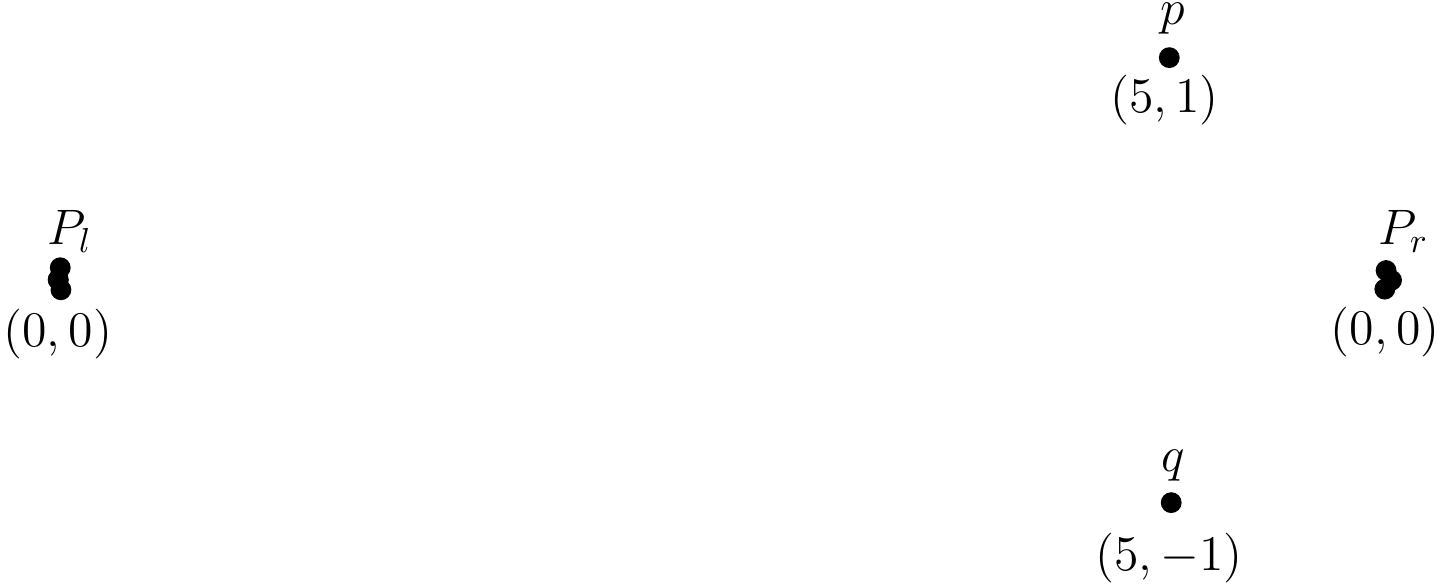}
    \caption{A set $P$ of $n=2m+2$ points in a convex position.}
    \label{fig:notPlanarPathNew}
\end{figure}

For simplicity of computation, we assume that a path connecting the points in $P_l$ has a Wiener index zero, and also 
a path connecting the points in $P_r$ has a Wiener index zero.
Thus, any path $\Pi$ of $P$ that aims to minimize the Wiener index will connect the points in $P_l$ by a path  and the points in $P_r$ by a path. 
We computed the Wiener index of all possible Hamiltonian paths defined on points $(0,0)$, $(6,0)$, $p$, and $q$; see Figure~\ref{fig:TwilveDiffConfig}.
This computation shows that the Hamiltonian path of the minimum  Wiener index is not planar (for sufficiently large $n$).~\qed
\begin{figure}[htp]
    \centering
        \includegraphics[width=0.9\textwidth]{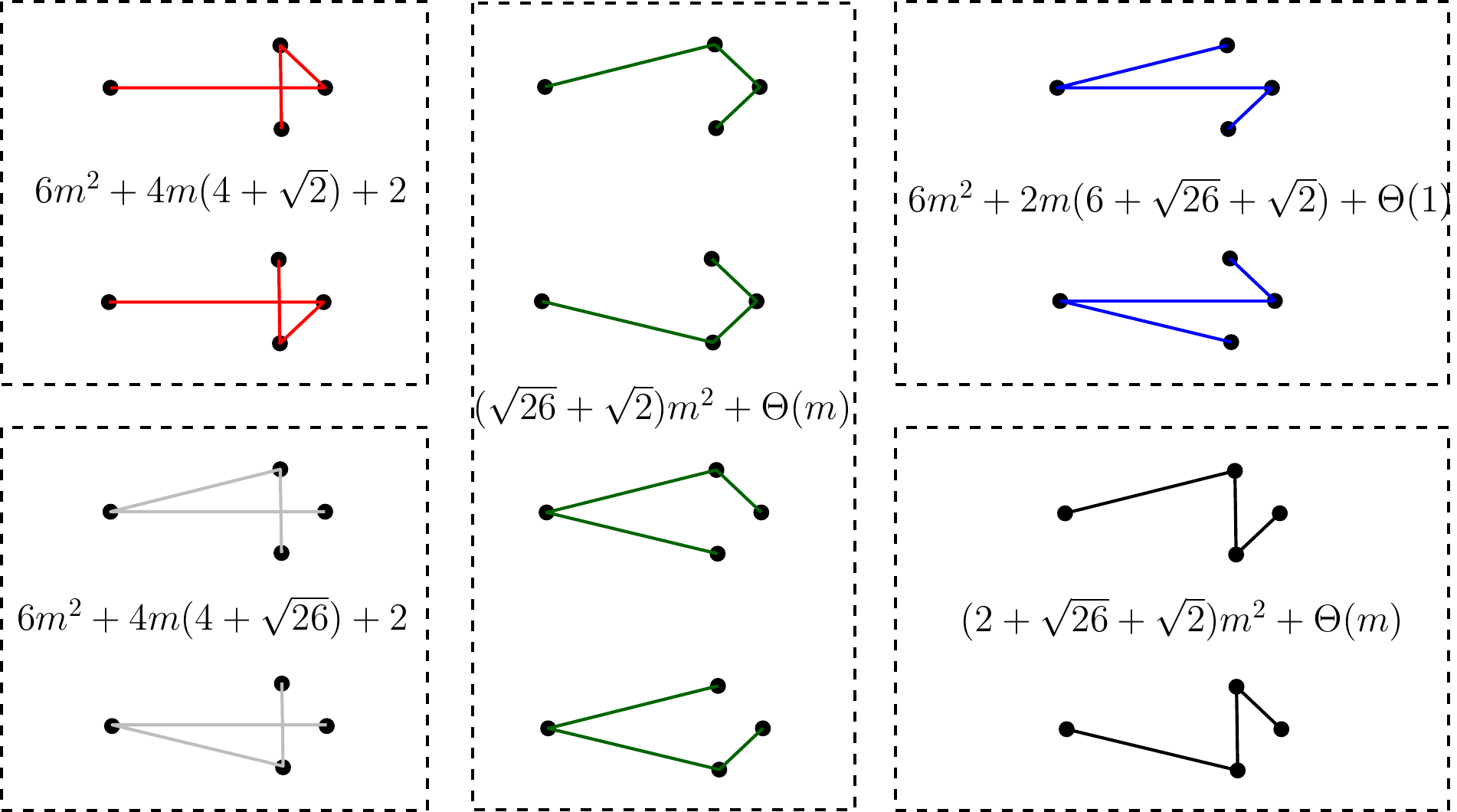}
    \caption{The Wiener index of the 12 possible Hamiltonian paths that are defined on points $(0,0)$, $(6,0)$, $p$, and $q$ (assuming that 
    the $m$ points on $(0,0)$ are connected by a path, and the $m$ points on $(6,0)$ are connected by a path, both of Wiener index zero). }
    \label{fig:TwilveDiffConfig}
\end{figure}

\end{proof}

\begin{theorem} 
For points in the Euclidean plane, it is NP-hard to compute a Hamiltonian path minimizing Wiener index. 
\end{theorem}
\begin{proof}
We reduce from Hamiltonicity in a grid graph (whose vertices are integer grid points and whose edges join pairs of grid points at distance one).  
First, observe that the Wiener index of a Hamiltonian path of $n$ points, where each edge is of length one, is $\sum_{i=1}^{n - 1} i (n - i) = { n+1 \choose 3}$; see Figure~\ref{fig:HamiltoninPathReduction}.
Thus, it is easy to see that a grid graph $G$ has a Hamiltonian path if and only if there exists a path of Wiener index ${ n+1 \choose 3}$.~\qed
\end{proof}

\begin{theorem}
There exists a set $P$ of $n$ points in the plane, such that the Wiener index of any Hamiltonian path is at least $\Theta(\sqrt{n})$ 
times the  Wiener index of the complete Euclidean graph over $P$
\end{theorem}
\begin{proof}
Let  $P$ be a set of $n$ points located on a$\sqrt{n}\times\sqrt{n}$ integer grid.
The Wiener index of any Hamiltonian path of $P$ is at least  ${ n+1 \choose 3}$, which is the Wiener index of a Hamiltonian path whose all its edges are of length one. Thus, the Wiener index of any Hamiltonian path of $P$ is at least $\Theta(n^{3})$. 
On the other hand, the Wiener index of the complete graph over $P$ is $\Theta(n^{2.5})$.~\qed
\end{proof}
\begin{figure}[H]
    \centering
        \includegraphics[width=0.75\textwidth]{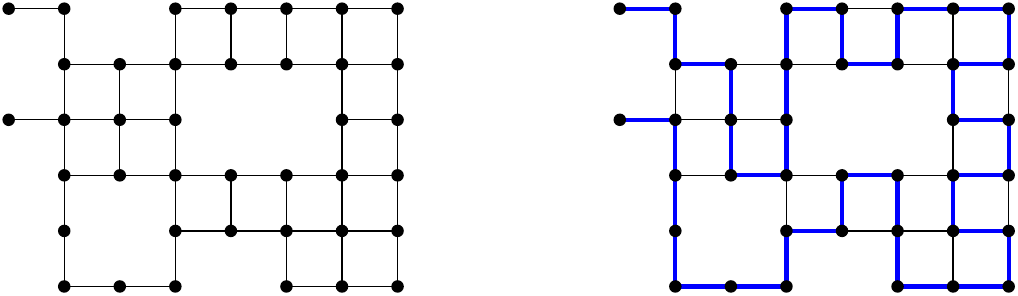}
    \caption{A grid graph $G$ and a Hamiltonian path with Wiener index  ${ n+1 \choose 3}$ in $G$.}
    \label{fig:HamiltoninPathReduction}
\end{figure}
\old{ NOTES TO OURSELVES

\section{Next Steps}

\subsection{Running Time}
Can the running time be improved?
(Prove conditional lower bounds?)

\subsection{More General Cases}
Generalize: Consider a set $S$ of $n$ points on the boundary of a simple polygon, $P$, having $m$ vertices, with the requirement to construct a tree within the polygon $P$ spanning all of $S$. Edges linking two points of $S$ will be geodesic paths within $P$.

Claim1: We still get planarity of an optimal Wiener index tree.

Claim2: Our DP still applies.  (running time now becomes what exactly?  I suspect $O(m+n^4)$.

Variant of this case: The visibility graph, $VG(S)$, is connected, and we are required to build the spanning tree within it. Then the running time will likely be something like $O(|E|^2)$, if $E$ is the set of edges in $VG(S)$.

More generally: What if $P$ has a small number, $h$, of holes? We can likely get $n^{O(h)}$, but might we be able to get FPT in $h$?

Related: What if the points are in the Euclidean plane, organized in $k$ subsets, each in convex position?  Again, $n^{O(k)}$ seems possible (several things to prove), but can we achieve FPT in $k$?

\subsection{$L_1$ Metric}

So far, we have been examining the Euclidean distance. The DP presented does not seem to care what the underlying metric is (Euclidean, $L_1$, geodesic distances with respect to a simple polygon $P$, etc).

But in the plane, does the minimum Wiener index spanning tree or path become easier to analyze and approximate in the $L_1$ metric?  I believe so.  And understanding the $L_1$ case could lead to improved approximation (to $\sqrt{2}$?) for the Euclidean case.
} 

\old{
\section{Concluding Remarks}

\old{
\section{Bounded Degree Tree} \label{Bounded Degree Tree}
 In this section, we show that for any set of points in the plane there exists a tree of bounded degree with Wiener index at most $(2+\epsilon)$ times the Wiener index of the complete graph.
 We show this by construction.

Let $P$ be a set of $n$ points in the plane, let $q$ be a point of $P$, and let $t > 1$ be a real number. A graph having the points of $P$ as its nodes is called a $q$-sink $t$-spanner for $P$, if for every point $p$ of $P$ there is a path in this graph of length at most $t |pq|$ from $p$ to $q$, where 
$|pq|$ is the Euclidean distance between $p$ and $q$.
In~\cite{narasimhanGeometricSpannerNetworks2007}, Narasimhan and Smid showed how to construct such a 
$q-sink$ $t-spanner$ in which each vertex has a bounded degree that only depends on $t$ (the maximal degree is $O(\frac{1}{(t-1)^2})$).  
  
Given a set of points $P$ and a constant $\epsilon > 0$, we construct the bounded degree tree for $P$  
whose Wiener index is at most $(2+\epsilon)$ the Wiener index of the complete graph of $P$ as follows:
\begin{enumerate}
    \item 
    Let $q = p_i \in P$, such that $w(S_i) = \sum\limits_{j = 1}^{n}|p_ip_j|$ is minimal, where
    $w(S_i)$ is the Wiener index of the star with the apex at $p_i$.
    \item
    Construct \emph{$q$-sink $t$-spanner} for $P$ with $t = 1+\frac{\epsilon}{2}$.
\end{enumerate}

\begin{claim}
     The resulted tree, $T$, is $(2+\epsilon)$-Wiener-approximation. 
 \end{claim}
 \begin{proof}
    The point $q$ is the point that minimizes the sum of distances to all the other points of $P$. 
    Next, let $p_1,p_2 \in P$:
    \\Clearly, $d_T(p_1,p_2) \leq d_T(p_1,q) +d_T(q,p_2)$ from triangle inequality.
    \\T is undirected $q-sink-(1+\frac{\epsilon}{2})-spanner$, therefore:
    \begin{equation*}
        \begin{aligned}
            d_T(p_1,q) +d_T(q,p_2) & \leq (1+\frac{\epsilon}{2})|p_1q|+(1+\frac{\epsilon}{2})|qp_2|\\
            & = (1+\frac{\epsilon}{2})\cdot d_{\str}(p_1,q) +(1+\frac{\epsilon}{2})\cdot d_{\str}(q,p_2)\\
            & = (1+\frac{\epsilon}{2})\cdot d_{\str}(p_1,p_2).
        \end{aligned}
    \end{equation*}
    Since the inequality holds for every $p_1,p_2 \in P$, it holds also to the sum of the distances, therefore
    \begin{equation*}
        \begin{aligned}
            W(T) & = \sum\limits_{\{u,v\} \subseteq P} d_{T}(u,v)\\
            & \leq \sum\limits_{\{u,v\} \subseteq P} (1+\frac{\epsilon}{2})\cdot d_{\str}(u,v)\\
            & = (1 + \frac{\epsilon}{2})W(\str).
        \end{aligned}
    \end{equation*}
    From claim \ref{1-star is 2-Wiener-approx} we know $W(\str) \leq 2 W(K_P)$, therefore
    \begin{equation*}
        \begin{aligned}
            W(T) & \leq (1 + \frac{\epsilon}{2})W(\str)\\
            & \leq (1 + \frac{\epsilon}{2})2W(K_P)\\
            & = (2+\epsilon)W(K_P).
        \end{aligned}
    \end{equation*}
    That is, $T$ is $(2+\epsilon)$-Wiener-approximation.~\qed
\end{proof}

In the constructed tree the maximal degree is $O(\frac{1}{(t-1)^2}) = O(\frac{1}{((1 + \frac{\epsilon}{2})-1)^2}) = O(\frac{1}{(\frac{\epsilon}{2})^2})$. 
}
} 

\bibliographystyle{plainurl}

\end{document}